\documentclass[11pt]{article}

\usepackage{amsmath}
\usepackage{amssymb}
\usepackage{amsfonts}
\usepackage{latexsym}
\usepackage{color}
\usepackage{float}
\usepackage{tikz}
\usepackage{graphicx,graphics,epstopdf}
\usepackage{amscd}
\usepackage{subfigure}
\catcode `\@=11 \@addtoreset{equation}{section}

\catcode `\@=12

\newcommand{\be}{\begin{equation}}
\newcommand{\en}{\end{equation}}
\newcommand{\bea}{\begin{eqnarray}}
\newcommand{\ena}{\end{eqnarray}}
\newcommand{\beano}{\begin{eqnarray*}}
\newcommand{\enano}{\end{eqnarray*}}
\newcommand{\bee}{\begin{enumerate}}
\newcommand{\ene}{\end{enumerate}}

\newcommand{\mb}{\mathbb}

\newcommand{\R}{\mathbb{R}}
\newcommand{\mc}{\mathcal}

\newcommand{\D}{{\mc D}}

\newcommand{\Sc}{{\cal S}}

\newcommand{\F}{{\cal F}}

\newcommand{\Lc}{{\cal L}}

\newcommand{\1}{1 \!\! 1}

\newcommand{\Hil}{\mc H}

\newtheorem{thm}{Theorem}
\newtheorem{cor}[thm]{Corollary}

\newtheorem{prop}[thm]{Proposition}
\newtheorem{defn}[thm]{Definition}

\newenvironment{proof}{\noindent {\bf Proof --}}{\hfill$\square$ \vspace{3mm}\endtrivlist}
\newcommand{\SC}{\mathcal{S}({\mb R})}
\newcommand{\SSS}{\mathcal{S}'({\mb R})}

\newcommand{\ip}[2]{\langle {#1},{#2}\rangle}

\catcode `\@=11 \@addtoreset{equation}{section}
\catcode `\@=12

\begin{document}

\thispagestyle{empty}

\vspace*{2cm}

\begin{center}
{\Large \bf {Coordinate representation for non Hermitian position and momentum operators}}   \vspace{2cm}\\

{\large F. Bagarello}\\
 DEIM -Dipartimento di Energia, ingegneria dell' Informazione e modelli Matematici,
\\ Scuola Politecnica, Universit\`a di Palermo, I-90128  Palermo, Italy\\
and\\ INFN, Sezione di Napoli,\\
and\\ Department of Mathematics and Applied Mathematics, \\
University of Cape Town, South Africa\\
e-mail: fabio.bagarello@unipa.it

\vspace*{.5cm}

{\large F. Gargano, S. Spagnolo, S. Triolo}\\
  Dipartimento di Energia, Ingegneria dell'Informazione e Modelli Matematici,\\
Scuola Politecnica, Universit\`a di Palermo,\\ I-90128  Palermo, Italy\\

\end{center}

\vspace*{1cm}

\begin{abstract}
\noindent In this paper we undertake an analysis of the eigenstates of two non self-adjoint operators $\hat q$ and $\hat p$ similar, in a suitable sense, to the self-adjoint position and momentum operators $\hat q_0$ and $\hat p_0$ usually adopted in ordinary quantum mechanics. In particular we discuss conditions for these eigenstates to be {\em biorthogonal distributions}, and we discuss few of their properties. We illustrate our results with two examples, one in which the similarity map between the self-adjoint and the non self-adjoint is bounded, with bounded inverse, and the other in which this is not true. We also briefly propose an alternative strategy to deal with $\hat q$ and $\hat p$, based on the so-called {\em quasi *-algebras}.
\end{abstract}

\newpage

\section{Introduction}

Since some years Quantum Mechanics driven by non self-adjoint Hamiltonians with real eigenvalues has been more and more investigated by several authors, both from a physical point of view, \cite{BB98}-\cite{MLmiao1}, and with a more mathematical perspective, see \cite{petr}-\cite{int1} and the recent volume \cite{book}.
Quite often the interest is focused on the analysis of the eigenvalues and the eigenstates of some {\em Hamiltonian} operator $H$  which, contrarily to what is usually assumed in (ordinary) Quantum Mechanics, is different from $H^\dagger$. In many case, however, the eigenvalues of such an $H$ are still real, at least for a certain range of the parameters of the model, the {\em unbroken phase}, which differs from the {\em broken phase} since, in this case, some of the eigenvalues of $H$ can be complex, \cite{ben3,BG14}.

What is not discussed in the literature so far, in our knowledge, is the role of non self-adjointness in the main standard ingredients of $H$, i.e. in the position and in the momentum operators. In fact, in many physical systems the self-adjoint Hamiltonian $H_0$ is just some suitable function of the operators $\hat q_0=\hat q_0^\dagger$ and $\hat p_0=\hat p_0^\dagger$, satisfying $[\hat q_0,\hat p_0]=i\1$, where $\1$ is the identity operator. For instance, for the harmonic oscillator, we have $H_0=\frac{1}{2}(\hat p_0^2+\hat q_0^2)$. Therefore, it is surely interesting, and natural, to look also at the functional properties of both $\hat q_0$ and $\hat p_0$. And, in fact, these have been studied along the years by several authors. In particular, we refer to \cite{gieres} for a rather interesting and clear review on this topic. Recently, \cite{bagriferimento}, one of us has begun the analysis of some aspects of two operators, $\hat q$ and $\hat p$, which behave as a deformed version of $\hat q_0$ and $\hat p_0$, while they still keep some of their essential aspects. This was done in connection with bi-coherent states, in order to prove their completeness under suitable conditions.
Here, rather than on bi-coherent states, we shall work on the possibility of defining, in some mathematically rigorous manner, the eigenstates of $\hat q$ and $\hat p$, and we discuss some of their properties. We believe that this analysis can be relevant in the foregoing discussion on the coexistence of two (or more) non self-adjoint observables, see \cite{ZSRML17,fring} for instance.

 The paper is organized as follows.
 In Section \ref{sect2} we shall analyse in a rigorous way the family of eigenstates of the deformed position and momentum operators $\hat q$ and $\hat p$ by adopting a distributional approach which essentially extends the one in \cite{gieres}. We shall also briefly propose an alternative approach based on the so-called {\em quasi *-algebras}, \cite{ait,bagrev}. Section \ref{sect3} contains two examples of different nature: in the first one the operators $\hat q_0$ and $\hat p_0$ are related to $\hat q$ and $\hat p$ by a bounded operator $T$ with bounded inverse. In the second example $T$ is still bounded, but $T^{-1}$ is not. Of course, other possibilities are also possible, and some of them can be found in the literature, see \cite{baginbagbook}. For this reason, we believe that having a general settings where all these possibilities fit, can be important and useful.  Our conclusions are given in \ref{sectconcl}.

\section{Well-behaved sets of distributions}\label{sect2}

Let $\hat q_0$ and $\hat p_0$ the self-adjoint multiplication and momentum operators defined as follows:
$$
\hat q_0\,\varphi(x)=x\,\varphi(x),\qquad \hat p_0\,\varphi(x)=-i\varphi'(x),
$$
for all $\varphi(x)\in\SC$, the
Schwartz space of rapidly decreasing $C^\infty$-functions on ${\Bbb R}$. Notice that $\SC$ is not the maximal domain of these operators. In fact, these sets are
$$
D_{max}(\hat q_0)=\{f(x)\in\Lc^2(\Bbb R): \, xf(x)\in\Lc^2(\Bbb R)\}, \quad D_{max}(\hat p_0)=\{f(x)\in\Lc^2(\Bbb R): \, f'(x)\in\Lc^2(\Bbb R)\}.
$$
Of course, $\SC\subset D_{max}(\hat q_0)\cap D_{max}(\hat p_0)$, and is dense in $\Lc^2(\Bbb R)$. In the rest of the paper, also in view of its role in quantum mechanics, we will often use $\SC$ as a suitable domain for the operators we will introduce in our analysis, even in lack of its maximality. This is useful also because $\SC$ is stable under the action of $\hat q_0$ and $\hat p_0$, while $D_{max}(\hat q_0)$ and $D_{max}(\hat p_0)$ are not. In particular, $\SC\subseteq D^\infty(\hat q_0)\cap D^\infty(\hat p_0)$, where $D^\infty(X)=\cap_{k\geq0}D(X^k)$ is the domain of all the powers of the operator $X$. It is well known that neither $\hat q_0$ nor $\hat p_0$ admit square integrable eigenvectors, and that their spectra coincide with $\Bbb R$:
\be
\hat q_0\,\xi_{x_0}(x)=x_0\,\xi_{x_0}(x), \qquad \hat p_0\,\theta_{p_0}(x)=p_0\,\theta_{p_0}(x),
\label{11}\en
where $x_0$ and $p_0$ are real numbers,
and
\be
\xi_{x_0}(x)=\delta(x-x_0), \qquad \theta_{p_0}(x)=\frac{1}{\sqrt{2\pi}}\,e^{ip_0x}.
\label{12}\en
Of course, $\xi_{x_0}(x), \theta_{p_0}(x)\in \SSS$, the set of tempered distributions (i.e. the continuous linear functionals on $\SC$). Two well known properties of the generalized eigenvectors $\{\xi_{x_0}(x), \, x_0\in\Bbb R\}$, are the following equalities:
\be\left<\xi_{x_0},\xi_{y_0}\right>=\delta(x_0-y_0), \qquad
\int_{\Bbb R}dx_0\,|\xi_{x_0}\left>\right<\xi_{x_0}|=\1.
\label{add1}\en
Here $\left<\xi_{x_0},\xi_{y_0}\right>$ can be seen as an extension of the scalar product in $\Lc^2(\Bbb R)$ to two elements in $\SSS$, which gives as a result another element in $\SSS$. The way in which this extension is constructed works as follows: given two distributions $F$ and $G$ in $\SSS$, we define $\left<F,G\right>=(\overline F\ast\tilde G)(0)$\footnote{This is because of the following equality: $\left<\varphi,\psi\right>=(\overline \varphi\ast\tilde \psi)(0)$, which holds for each $\varphi(x),\psi(x)\in\SC$. Here $\tilde \psi(x)=\psi(-x)$.}, where $\tilde G(x)=G(-x)$, \cite{vladim}. In other words, we use the convolution between distributions to extend the scalar product to $\SSS$. The convolution $\overline F\ast\tilde G$ is further defined as follows:
\be
\left(\overline F\ast\tilde G,\varphi\right)=\int_{\Bbb R}\int_{\Bbb R}\overline{F(x)}\,\tilde G(y)\,\varphi(x+y)\,dx\,dy=\left<F,G\ast\varphi\right>,
\label{add3}\en
$\forall\varphi(x)\in\SC$. It is well known, \cite{vladim}, that the convolution between two distributions does not always exist.
However, sufficient conditions on the support of the distributions are considered in the literature which ensure the existence of $\overline F\ast\tilde G$, and therefore of $\left<F,G\right>$, and these conditions are satisfied in our case: if we take $F=\xi_{x_0}$ and $G=\xi_{y_0}$, simple computations show that $$\left(\overline \xi_{x_0}\ast\tilde \xi_{y_0},\varphi\right)=\left<\xi_{x_0},\xi_{y_0}\ast\varphi\right>=\int_{\Bbb R}\xi_{x_0}(x)\varphi(x-y_0)dx=\varphi(x_0-y_0)=\left(\xi_{t_0},\varphi\right),$$
where $t_0=x_0-y_0$, for all $\varphi(x)\in\SC$. Hence $\left(\overline \xi_{x_0}\ast\tilde \xi_{y_0}\right)(x)=\xi_{t_0}(x)$, and therefore
$$
\left<\xi_{x_0},\xi_{y_0}\right>=\left(\overline \xi_{x_0}\ast\tilde \xi_{y_0}\right)(0)=\xi_{t_0}(0)=\delta(x_0-y_0),
$$
which is the first equality in (\ref{add1}). The second equality follows from the fact that, for all $\varphi(x)\in\SC$, $\varphi(x_0)=\left<\xi_{x_0},\varphi\right>$. Then we have
$$
\varphi(x)=\int_{\Bbb R}\delta(x-x_0)\varphi(x_0)\,dx_0=\int_{\Bbb R}\xi_{x_0}(x)\left<\xi_{x_0},\varphi\right>\,dx_0,
$$
as we had to prove.
Incidentally, we see that the resolution of the identity in (\ref{add1}) has to be intended (at least) on $\SC$.

Let us now consider an operator $T$, not necessarily bounded, with domain $D(T)$ larger than (or equal to) $\SC$. Then:

\begin{defn} \label{def1}
$T$ is $\SC$-stable if $T$ is invertible and if $T$, $T^{-1}$, $T^\dagger$ and $(T^{-1})^\dagger=(T^\dagger)^{-1}$ all map $\SC$ in $\SC$. Moreover, an $\SC$-stable operator $T$ is called fully $\SC$-stable if $T^\dagger$ and $T^{-1}$ map $\SC$ into itself continuously.

\end{defn}

Examples of these operators will be given later in the paper: essentially, a fully $\SC$-stable operator $T$ leaves $\SC$ stable, together with its adjoint, its inverse, and the inverse of its adjoint, and is such that, if $\{\varphi_n(x)\}$ is a sequence in $\SC$ which converges to $\varphi(x)$ in the $\tau_\Sc$ topology, then both $\{T^{-1}\varphi_n(x)\}$ and $\{T^\dagger\varphi_n(x)\}$ converge in the same topology. Each $\SC$-stable operator can be used to construct a non self-adjoint version of $\hat q_0$ and $\hat p_0$, as we will show now:

let $T$ be a  $\SC$-stable operator. Then, two operators $\hat q$ and $\hat p$ can be defined as follows:
\be
\hat q \varphi=T\hat q_0T^{-1}\varphi,\qquad \hat p\varphi=T\hat p_0T^{-1}\varphi,
\label{13}\en
for all $\varphi(x)\in\SC$. Of course, $\hat q$ and $\hat p$ map $\SC$ in $\SC$, so that they are, in particular, densely defined. As for $\hat q_0$ and $\hat p_0$, we are not interested here in $D_{max}(\hat q)$ and $D_{max}(\hat p)$, even if these sets could be larger than $\SC$. It is also possible to check that their adjoints satisfy the following:
\be
\hat q^\dagger \varphi=(T^{-1})^\dagger\hat q_0T^{\dagger}\varphi,\qquad \hat p^\dagger \varphi=(T^{-1})^\dagger\hat p_0T^{\dagger}\varphi,
\label{14}\en
for all $\varphi(x)\in\SC$. Hence, also $\hat q^\dagger, \hat p^\dagger$ leave $\SC$ stable.

\vspace{2mm}

{\bf Remark:--} A {\em physical way} to understand these results is by noticing that $\hat q_0$, $\hat q$ and $\hat q^\dagger$ are similar or, in other words, that they satisfy suitable intertwining relations. When this happens for (bounded) operators with discrete spectra, this implies that their eigenvectors are related by the intertwining operators ($T$, in our case).  Of course, the operators $\hat p_0$, $\hat p$ and $\hat p^\dagger$ are also similar.

\vspace{2mm}

Now, we want to  construct the eigenvectors of $\hat q$ and $\hat q^\dagger$, as well as those of $\hat p$ and $\hat p^\dagger$, and to check that they  produce a family of  {\em well-behaved} states, in the sense of \cite{bagriferimento}. We recall that, if $x\in\Bbb R$ labels a tempered distribution $\eta_x\in \Sc'(\Bbb R)$,  and if $\F_\eta$ is the set of all these distributions, $\F_\eta=\{\eta_x,\,x\in\Bbb R\}$, then

\begin{defn}\label{def2}

 $\F_\eta$ is called {$\hat q$-\em well-behaved} (or, simply, {\em well-behaved}) if:

 1. each $\eta_x$ is a generalized eigenstate of $\hat q$: $\hat q\,\eta_x=x\,\eta_x$, for all $x\in\Bbb R$;

 2. a second family of generalized vectors $\F^\eta=\{\eta^x\in\Sc'(\Bbb R),\,x\in \Bbb R\}$ exists such that $\left<\eta_x,\eta^y\right>=\delta(x-y)$ and $\int_{\Bbb R}dx|\eta_x\left>\right<\eta^x|=\int_{\Bbb R}dx|\eta^x\left>\right<\eta_x|=\1$, at least on $\SC$.

\end{defn}

To do that, we first need to extend $T$ and $(T^{-1})^\dagger$ to all of $\SSS$ even if, in principle, it would be sufficient to extend them to $\xi_{x_0}(x)$. This would be enough to define the eigenstates of $\hat q$ and $\hat q^\dagger$. But, since we also want to construct eigenstates of $\hat p$ and $\hat p^\dagger$, we prefer to set us a more general approach.

First of all it is clear that $T$ and $(T^{-1})^\dagger$ can be extended to $\SSS$ by duality. In fact, we define
\be
\left<TF,\varphi\right>=\left<F,T^\dagger\varphi\right>, \qquad \left<(T^{-1})^\dagger F,\varphi\right>=\left<F,T^{-1}\varphi\right>,
\label{15}\en
for all $F\in\SSS$ and $\varphi\in\SC $. Here $\left<.,.\right>$ is the form which puts in duality $\SC$ and $\SSS$, which extends the standard scalar product in  $\Lc^2(\Bbb R)$\footnote{This form can be defined as before, via convolution of distributions.}. It is obvious that formulas in (\ref{15}) make sense, since $\SC$ is stable under the action of $T^\dagger$ and $T^{-1}$. It is also clear that $TF$ and $(T^{-1})F$ define  linear functionals on $\SC$. To conclude that they are both tempered distributions, we still have to check that they are both $\tau_\Sc$-continuous, i.e. that if $\{\varphi_n(x)\in\SC\}$ is a sequence $\tau_\Sc$-convergent to $\varphi(x)\in\SC$, then
$
\left<TF,\varphi_n\right>\rightarrow \left<TF,\varphi\right>$, and  $\left<(T^{-1})^\dagger F,\varphi_n\right>\rightarrow \left<(T^{-1})^\dagger F,\varphi\right>,$
for all $F\in\SSS$.
This is certainly true if $T$ is fully $\SC$-stable since, for instance, we have
$$
\left<TF,\varphi_n\right>=\left<F,T^\dagger\varphi_n\right>\rightarrow\left<F,T^\dagger\varphi\right>=\left<TF,\varphi\right>,
$$
due to the fact that, if $\varphi_n(x)$ converges in the topology $\tau_S$, then $(T^\dagger\varphi_n)(x)$ converges as well, in the same topology. Then the following Corollary holds:

\begin{cor}\label{cor3}
Let $T$ be a fully $\SC$-stable operator. Then
\be
\eta_{x_0}(x)=(T\xi_{x_0})(x), \qquad \eta^{x_0}(x)=((T^{-1})^\dagger\xi_{x_0})(x),
\label{16}\en
are tempered distributions. Moreover, $\eta_{x_0}(x)\in D(\hat q)$ and $\eta^{x_0}(x)\in D(\hat q^\dagger)$, and we have:
\be
(\hat q\,\eta_{x_0})(x)=x_0\,\eta_{x_0}(x),\qquad (\hat q^\dagger \eta^{x_0})(x)=x_0\,\eta^{x_0}(x)
\label{17}\en
\end{cor}

\vspace{2mm}

{\bf Remark:--} We recall that, in general, the domain of a given operator $X$ acting on $\Hil$, $D(X)$, is the set of the following subset of $\Hil$: $D(X)=\{f\in\Hil:\, Xf\in\Hil\}$. Here we are slightly extending this notion, while keeping the same notation. In fact, $\eta_{x_0}(x)\in D(\hat q)$ would not be compatible with the above definition, since $\eta_{x_0}(x)$ does not belong to $\Lc^2(\Bbb R)$. For this reason $D(\hat q)$ and $D(\hat q^\dagger)$ should be understood as {\em generalized domains}: $D(\hat q)=\{F\in\SSS:\:\, \hat qF\in\SSS\}$, $D(\hat q^\dagger)=\{F\in\SSS:\:\, \hat q^\dagger F\in\SSS\}$, and so on.

\vspace{2mm}

In the same way, out of $(\hat p_0,\theta_{p_0}(x))$ in \eqref{11}-\eqref{12}, we can deduce the eigenstates of the operators $\hat p$ and $\hat p^\dagger$ introduced in (\ref{13}) and (\ref{14}), as in (\ref{16}):
\be
\mu_{p_0}(x)=(T\theta_{p_0})(x), \qquad \mu^{p_0}(x)=((T^{-1})^\dagger\theta_{p_0})(x),
\label{18}\en
and they are both tempered distributions, belong respectively to $D(\hat p)$ and $D(\hat p^\dagger)$ (in the sense of the previous Remark) and are generalized eigenstates of $\hat p$ and $\hat p^\dagger$ respectively:
\be
(\hat p\mu_{p_0})(x)=p_0\,\mu_{p_0}(x),\qquad (\hat p^\dagger \mu^{p_0})(x)=p_0\,\mu^{p_0}(x).
\label{19}\en

Let us now introduce the following sets of tempered distributions: $\F_\eta=\{\eta_{x_0}(x), \,x_0\in\Bbb R\}$, $\F^\eta=\{\eta^{x_0}(x), \,x_0\in\Bbb R\}$, $\F_\mu=\{\mu_{p_0}(x), \,p_0\in\Bbb R\}$ and $\F^\mu=\{\mu^{p_0}(x), \,p_0\in\Bbb R\}$.

For our purposes, it is convenient now (but not strictly necessary) to assume some extra properties to $(T^{-1})^\dagger$. In particular we assume that
  \be
  \left((T^{-1})^\dagger\xi_{y_0}\right)\ast\varphi=(T^{-1})^\dagger\left(\xi_{y_0}\ast\varphi\right),
\label{add2}\en
for all $\varphi(x)\in\SC$. Here $\ast$ is the convolution. Notice that both sides of this equality are well defined. In fact, we have already seen that $\eta^{y_0}=(T^{-1})^\dagger\xi_{y_0}$ is in $\SSS$, and therefore it admits a convolution with $\varphi(x)$, since this function belongs to $\SC$. The result is, in general, a $C^\infty$ function increasing not faster than some polynomial. As for the right-hand of (\ref{add2}), we see that $(\xi_{y_0}\ast\varphi)(x)=\int_{\Bbb R}\xi_{y_0}(s)\varphi(x-s)ds=\int_{\Bbb R}\delta(s-y_0)\varphi(x-s)ds=\varphi(x-y_0)$, which belongs to $\SC$. Hence, we can act on this function with $(T^{-1})^\dagger$. However, it is not granted a priori that the sides of (\ref{add2}) coincide. However, we observe that (\ref{add2}) is very similar to some well known properties of convolutions of distributions, like, for instance, the property that $(F\ast G)'=F'\ast G$, for all distributions $F$ and $G$ for which their convolution exists. We refer to \cite{vladim} for many details on distribution theory, and for more situations in which (\ref{add2}) is again satisfied.

Now we have

\begin{prop}\label{prop11} Under the assumption (\ref{add2}) the set $\F_\eta$ is well behaved on $\SC$.

\end{prop}

\begin{proof}
Let us take $\varphi(x),\psi(x)\in\SC$. Then we have, using the stability of $\SC$ under the action of both $T^\dagger$ and $T^{-1}$, as well as the resolution of the identity in (\ref{add1}), valid for all functions in $\SC$,
$$
\int_{\Bbb R}dx_0\left<\varphi,\eta_{x_0}\right>\left<\eta^{x_0},\psi\right>=\int_{\Bbb R}dx_0\left<\varphi,T\xi_{x_0}\right>\left<(T^{-1})^\dagger\xi_{x_0},\psi\right>=$$
$$=\int_{\Bbb R}dx_0\left<T^\dagger\varphi,\xi_{x_0}\right>\left<\xi_{x_0},T^{-1}\psi\right>=\left<T^\dagger\varphi,T^{-1}\psi\right>=
\left<\varphi,\psi\right>.
$$
In a similar way we get
$$
\int_{\Bbb R}dx_0\left<\varphi,\eta^{x_0}\right>\left<\eta_{x_0},\psi\right>=\int_{\Bbb R}dx_0\left<\varphi,(T^{-1})^\dagger\xi_{x_0}\right>\left<T\xi_{x_0},\psi\right>=
\left<\varphi,\psi\right>.
$$
To prove now that $\left<\eta_{x_0},\eta^{y_0}\right>=\delta(x_0-y_0)$ we first recall that, see (\ref{add3}), $
\left(\overline \eta_{x_0}\ast\tilde \eta^{y_0},\varphi\right)=\left<\eta_{x_0},\eta^{y_0}\ast\varphi\right>$, for all $\varphi(x)\in\SC$. But, since $(\eta^{y_0}\ast\varphi)(x)=\varphi(x-y_0)$ is in $\SC$, using (\ref{add2}),
$$
\left<\eta_{x_0},\eta^{y_0}\ast\varphi\right>=\left<T\xi_{x_0},\eta^{y_0}\ast\varphi\right>=
\left<\xi_{x_0},T^\dagger\left(\eta^{y_0}\ast\varphi\right)\right>=
\left<\xi_{x_0},\left(\xi_{y_0}\ast\varphi\right)\right>=\left(\overline \xi_{x_0}\ast\tilde \xi_{y_0},\varphi\right).
$$
Hence $\left(\overline \eta_{x_0}\ast\tilde \eta^{y_0},\varphi\right)(x)=\left(\overline \xi_{x_0}\ast\tilde \xi_{y_0},\varphi\right)(x)$, and therefore we have
$$
\left<\eta_{x_0},\eta^{y_0}\right>=\left(\overline \eta_{x_0}\ast\tilde \eta_{y_0},\varphi\right)(0)=\left(\overline \xi_{x_0}\ast\tilde \xi_{y_0},\varphi\right)(0)=\left<\xi_{x_0},\xi_{y_0}\right>=\delta(x_0-y_0),
$$
which is what we had to prove.

\end{proof}

\vspace{2mm}

{\bf Remarks:--} (1) A similar proof can be repeated to prove that also $\F_\mu$ is $\hat p$-well-behaved.

(2) It might be interesting to observe that the equality $\left<\eta_{x_0},\eta^{y_0}\right>=\delta(x_0-y_0)$ can be deduced following a very different approach from the one adopted in Proposition \ref{prop11}, i.e. by making use of the so-called quasi-bases, see \cite{baginbagbook}. To construct the quasi-bases we consider an o.n. basis for $\Lc^2(\Bbb R)$, say $\F=\{e_n(x)\in\SC\}$. To be concrete, we can think of $\F_e$ as the set of eigenstates of the harmonic oscillator, since these are all in $\SC$, and form an o.n. basis of $\Lc^2(\Bbb R)$. Then we introduce $\F_\varphi=\{\varphi_n(x)=T e_n(x)\}$ and $\F_\Psi=\{\Psi_n(x)=(T^{-1})^\dag e_n(x)\}$. Assuming that $T$ is $\SC$-stable, then the functions $\varphi_n(x)$ and $\Psi_n(x)$ are all in $\SC$. Also, $\F_\varphi$ and $\F_\Psi$ are biorthogonal, $\left<\varphi_n,\Psi_m\right>=\delta_{n,m}$, and they are $\SC$-quasi bases, \cite{baginbagbook}. Indeed, it is easy to check that, for all $\gamma(x), \eta(x)\in\SC$, we have
\bea
\left<\gamma,\eta\right>=\sum_n\left<\gamma,\varphi_n\right>\left<\Psi_n,\eta\right>= \sum_n\left<\gamma,\Psi_n\right>\left<\varphi_n,\eta\right>\label{pers}.
\ena
Our working assumption here is that this same equality can be extended outside $\SC$, to all distributions $\xi_{x_0}$, $x_0\in\Bbb R$. Indeed, assuming for instance that
\bea
\delta(x_0-y_0)=\left<\xi_{x_0},\xi_{y_0}\right>=\sum_n\left<\xi_{x_0},\varphi_n\right>\left<\Psi_n,\xi_{y_0}\right>,\label{xiproduct}
\ena
it is easy to conclude, again, that
\bea
\left<\eta_{x_0},\eta^{y_0}\right>=\left<\eta^{y_0},\eta_{x_0}\right>=\delta(x_0-y_0).
\ena
In fact we have
\beano\left<\eta^{y_0},\eta_{x_0}\right>=
\sum_n\left<(T^{-1})^\dag\xi_{y_0},\varphi_n\right>\left<\Psi_n,T^{-1}\xi_{y_0}\right>=
\sum_n\left<\xi_{x_0},T^{-1}\varphi_n\right>\left<T^{\dag}\Psi_n,\xi_{y_0}\right>=\\
\sum_n\left<\xi_{x_0},e_n\right>\left<e_n,\xi_{y_0}\right>=
\sum_ne_n(x_0)\overline{e_n(y_0)}=\delta(x_0-y_0).
\enano

 This approach is particularly interesting since it is heavily connected with the general settings proposed in recent years for deformed canonical commutation and anti-commutation relations, see \cite{BP2016,BGV15,BP2017} for  recent applications, and since makes no use of equality (\ref{add2}), which is not always satisfied even in simple cases, as in the first example discussed in Section \ref{Ex1}.

\vspace{2mm}

Following \cite{bagriferimento}, we can also introduce now two operators, $S_\eta$ and $S^\eta$, on the following generalized domains:
$$
D(S_\eta)=\left\{F\in\SSS:\int_{\Bbb R}dx\left<\eta_x,F\right>\eta_x\in\SSS\right\},
$$
$$D(S^\eta)=\left\{G\in\SSS:\int_{\Bbb R}dx\left<\eta^x,G\right>\eta^x\in\SSS\right\},
$$
and
\be
S_\eta F=\int_{\Bbb R}dx\left<\eta_x,F\right>\eta_x, \qquad S^\eta G=\int_{\Bbb R}dx\left<\eta^x,G\right>\eta^x,
\label{110}\en
for all $F\in D(S_\eta)$ and $G\in D(S^\eta)$. In particular it is clear that $\eta^y\in D(S_\eta)$ and that $\eta_y\in D(S^\eta)$, for all $y\in\Bbb R$. In particular, $S_\eta\eta^y=\eta_y$, while $S^\eta\eta_y=\eta^y$. Furthermore, if $T$ is bounded, then for all $\varphi(x)\in\SC$ we get $S_\eta \varphi=TT^\dagger \varphi$. Also, if $T^{-1}$ is bounded, then $S^\eta f=(T^{-1})^\dagger T^{-1}\varphi$. Of course, when they are both bounded, we see that $S_\eta$ and $S^\eta$ are one the inverse of the other. More results on $S_\eta$ and $S^\eta$ are discussed in \cite{bagriferimento}, where some connections of these operators with the so-called $kq-$representation, see e.g. \cite{zak2}, are also considered. Here we just want to notice that similar operators are somehow used in the literature to define new scalar products in the Hilbert space, see \cite{baginbagbook} and references therein.

\subsection{A brief algebraic view to $\hat q$ and $\hat p$}

In what we have done so far, we have used techniques borrowed from
functional analysis and distribution theory to deal with $\hat q$,
$\hat p$, and their adjoints. Now, we briefly suggest a possible
alternative approach to deal with these operators, based on
certain algebras of unbounded operators. We refer to \cite{ait}
for a mathematically oriented monograph, and to \cite{bagrev} for
a more physically focused review.

If $\D$ is a dense subspace of a (separable) Hilbert space $\Hil$ we denote by
$\Lc^\dagger(\D)$ the set of all the operators  which leave,
together with their adjoints,  $\D$ invariant. Then
$\Lc^\dagger(\D)$ is a *-algebra with respect to the usual
operations. In particular, the adjoint in $\Lc^\dagger(\D)$ is just the restriction of the usual adjoint to $\D$. It is worth remarking that $\Lc^\dagger(\D)$ contains suitable unbounded operators, and this is, in fact, its main {\em raison d'etre}, \cite{ait,bagrev}. This can be easily understood since, in many concrete applications, $\D$ is taken to be the domain of all the powers of some suitable unbounded, densely defined, self-adjoint operator $N$ on $\Hil$: $\D=D^\infty(N)=\cap_{k\geq0}D(N^k)$, which, due to the assumptions on $N$, is automatically dense in $\Hil$. In particular, if $N=p^2+x^2$, where $p=-i\frac{d}{dx}$, it is known that $\D=\SC$, and that the topology $\tau_\Sc$ is equivalent to the one defined by the seminorms $
f\rightarrow \|N^nf\|, \qquad n\geq 0$,
\cite{pip}. Then we get the following {\em rigged Hilbert space}:
$$ \SC \subset L^2({\mb R}) \subset \SSS, $$
see also \cite{blt}. From now on we identify $\D$ with $\SC$. Therefore, the set of the distributions $\SSS$ is just the dual of $\D$.

\begin{defn}\label{deffa}
 An invertible operator $T$ in the Hilbert Space $\Hil$, such that $T, T^{-1}\in \Lc^\dagger(\D)$ is called  {\em admissible}  if
there exist an o.n. basis $\F_e=\{e_n(x)\in\SC\}$  for $\Lc^2(\Bbb R)$, such that  $\F_\varphi=\{\varphi_n(x)=T e_n(x)\}$ and $\F_\Psi=\{\Psi_n(x)=(T^{-1})^\dag e_n(x)\}$ are $\D'$ quasi bases, in the following sense: for every $F,G\in\D'$  \be
\left<F,G\right>=\sum_n\left<F,\varphi_n\right>\left<\Psi_n,G\right>=\sum_n\left<F,\Psi_n\right>\left<\varphi_n,G\right>\label{enz}.
\en

\end{defn}

\vspace{2mm}

{\bf Remark:--} Because of the properties of $\Lc^\dagger(\D)$, \cite{ait}, an element $A\in \Lc^\dagger(\D)$ is automatically {\em fully-admissible}. By this we mean that, for all sequence $\varphi_n(x)\in\SC$ $\tau_\Sc$-converging to $\varphi(x)$ the sequences $(T^\dagger \varphi_n)(x)$ and $(T^{-1} \varphi_n)(x)$ both converge in the same topology. It is evident how the concept of fully-admissibility can be seen as an {\em algebraic counterpart} of the fully stability we have introduced before. In fact, what we are doing in this short section, is just adopting a different language to deduce the same results.

\vspace{2mm}

Let now  $A\in \Lc^\dagger(\D)$  be  admissible, and therefore fully-admissible. Then $A$ can be
extended, by duality, to a continuous operator $A_{ex}:\SSS
\to \SSS$. In fact, to keep the notation simple and since no confusion can arise, in the following we identify $A$ and $A_{ex}$. We have: $ \left< {A}F, \varphi \right>=\left< F, A^\dag \varphi\right>$, $\forall F \in \SSS,\, \varphi\in \SC$.
This operator is still linear and it is also continuous: in fact, due to Definition \ref{deffa}, if $\varphi_n(x)\rightarrow \varphi(x)$ in the topology $\tau_\Sc$, then $A^\dagger \varphi_n\rightarrow A^\dagger \varphi$ in the same topology. Hence
$
\left<AF,\varphi_n\right>=\left<F,A^\dagger \varphi_n\right> \rightarrow \left<F,A^\dagger \varphi\right> = \left<A F, \varphi\right>,
$
for all $\varphi(x)\in \SC$  and $F\in \SSS$. In particular, then, $\eta_{x_0}=T\xi_{x_0}$ and $\eta^{x_0}=(T^{-1})^\dagger \xi_{x_0}$ are both in $\SSS$, and the following hold:
$$
\left<\eta_{x_0},\varphi\right>=(T^\dagger \varphi)(x_0), \qquad \left<\eta^{x_0},\varphi\right>=(T^{-1} \varphi)(x_0),
$$
for all $\varphi(x)\in\SC$.
Moreover, using again Definition \ref{deffa} and the fact that $\hat q_0\in \Lc^\dagger(\D)$, it is clear that $\hat q=T\hat q_0 T^{-1}$ is also in $\Lc^\dagger(\D)$.
Then we have:

%
%
%
%
%
%

\begin{prop} If  $T\in \Lc^\dagger(\D)$ is  admissible,  the set $\F_\eta=\{\eta_{x_0},\,x_0\in\Bbb R \}$ is well behaved.

\end{prop}

\begin{proof}
\begin{enumerate}
\item  For every $ x_0\in\R, \varphi(x) \in\SC$
$$\left< \hat {q} \eta_{x_0}, \varphi \right>:=\left<(T\hat q_0T^{-1})(T\xi_{x_0}),\varphi \right>=\left<T\hat q_0\xi_{x_0},\varphi \right>=\left<x_0T\xi_{x_0}, \varphi \right>=\left<x_0\eta_{x_0}, \varphi \right>.$$
Hence $\hat q \eta_{x_0}=x_0\eta_{x_0}$.
\item
\beano\left<\eta^{y_0},\eta_{x_0}\right>=
\sum_n\left<(T^{-1})^\dag\xi_{y_0},\varphi_n\right>\left<\Psi_n,T^{-1}\xi_{y_0}\right>=
\sum_n\left<\xi_{x_0},T^{-1}\varphi_n\right>\left<T^{\dag}\Psi_n,\xi_{y_0}\right>=\\
=\sum_n\left<\xi_{x_0},e_n\right>\left<e_n,\xi_{y_0}\right>=
\sum_ne_n(x_0)\overline{e_n(y_0)}=\delta(x_0-y_0).
\enano
\item For every $ x_0,y_0 \in\R$, and for all $\varphi(x),\psi(x)\in\SC$,
using the resolution of the identity in (\ref{add1}), valid for
all functions in $\SC$,
$$
\int_{\Bbb
R}dx_0\left<\varphi,\eta_{x_0}\right>\left<\eta^{x_0},\psi\right>=
\int_{\Bbb
R}dx_0\left<\varphi,T\xi_{x_0}\right>\left<(T^{-1})^\dagger\xi_{x_0},\psi\right>=\int_{\Bbb
R}dx_0\left<T^\dagger\varphi,\xi_{x_0}\right>\left<\xi_{x_0},T^{-1}\psi\right>=$$
$$=\left<T^\dagger\varphi,T^{-1}\psi\right>=
\left<\varphi,\psi\right>.
$$
\end{enumerate}
\end{proof}

Of course, a similar procedure can be repeated for $\hat p=T\hat p_0 T^{-1}$.

The conclusion of this analysis is that we could use the algebraic settings provided by $\Lc^\dagger(\D)$ rather than the one adopted previously by simply replacing the notion of stability with that of admissibility.  A deeper analysis of this alternative approach is postponed to a future paper.

\section{Examples}\label{sect3}

This section is devoted to the discussion of two examples of our general results. In particular, we will first analyze a situation in which both $T$ and $T^{-1}$ are bounded, and then a different situation in which $T$ is bounded, but $T^{-1}$ is not.

\subsection{First example}
\label{Ex1}
For every
$u,v \in \SC$  such that $\ip{u}{v}=1,$ we define the
operator $P_{u,v}f:=\ip{u}{f}v$. {Notice that $u$ and $v$ cannot have different parity, since in this case they would be automatically orthogonal. Then, either they have the same parity, or their parity is not defined}. Assume that $\alpha,\beta$ are complex numbers such that $\alpha+\beta+\alpha\beta=0$. Then, if $\alpha\neq-1$, $\beta=\frac{-\alpha}{1+\alpha}$, and the new operator
$$T=\1+\alpha P_{u,v}$$  is invertible, with $T^{-1}=\1+\beta P_{u,v}$.
Unless $u=v$ and $\alpha\in\R$, $T$ is not Hermitian, nor unitary, and we have $T^{\dag}=\1+\overline{\alpha}P_{v,u}\neq
T^{-1}$. Then $(T^{-1})^\dag=\1+\overline{\beta} P_{v,u}=(T^\dagger)^{-1}$.

Recalling that  $u,v \in \Sc({\mb R})$, $T$ turns out to be fully $\SC$-stable. In fact, first of all it is evident that $T, T^{-1},T^\dag,(T^{-1})^\dag$ all map $\SC$ in $\SC$.
Moreover, if $\{\varphi_n\in\SC\}$ is a sequence $\tau_\Sc$-convergent to $\varphi\in\SC$ then, for each $F\in\Sc'(\Bbb R)$,
$$
\left<F,T^\dag \varphi_n\right>=\left<F,\varphi_n+\overline{\alpha}\left<v,\varphi_n\right>u\right>=
\left<F,\varphi_n\right>+\overline{\alpha}\left<v,\varphi_n\right>\left<F,u\right>\longrightarrow \left<F,\varphi\right>+\overline{\alpha}\left<v,\varphi\right>\left<F,u\right>=$$
$$=
\left<F,\varphi+\overline{\alpha}\left<v,\varphi\right>u\right>=\left<F,T^\dag \varphi\right>.
$$
Similarly,
$
\left<F,T^{-1} \varphi_n\right>\rightarrow\left<F,T^{-1} \varphi\right>,
$
and therefore both $T^\dag$ and $T^{-1}$ map $\SC$ into itself with continuity.
As a consequence of the full $\SC$-stability of  $T$,  Corollary \ref{cor3} implies that, for each $x_0,x\in \R$,
$\eta_{x_0}(x),\eta^{x_0}(x)\in\SSS$, and, in particular, that  $\eta_{x_0}(x)\in D(\hat q)$ and $\eta^{x_0}(x)\in D(\hat q^\dag)$.
The following expressions for $\eta_{x_0}(x),\eta^{x_0}(x)$ follows from \eqref{16}:

\bea
\eta_{x_0}(x)&=&(T\xi_{x_0})(x)=\xi_{x_0}(x)+\alpha \ip{u}{\xi_{x_0}}v(x)=\xi_{x_0}(x)+\alpha\, v(x)\overline{u(x_0)},\label{21}\\
\eta^{x_0}(x)&=&\left((T^{-1})^\dagger\xi_{x_0}\right)(x)=\xi_{x_0}(x)+\overline{\beta \,v(x_0)}\,u(x)\label{22}.
\ena

We now prove that the sets $\F_\eta=\{\eta_{x_0}\in\Sc'(\Bbb R), \,x_0\in\Bbb R\}$ and $\F^\eta=\{\eta^{x_0}\in\Sc'(\Bbb R),\,x_0\in \Bbb R\}$ form two families of well-behaved states in the sense of Definition \ref{def2}.

In fact, we obtain the following:
\begin{enumerate}
\item From Corollary \ref{cor3}, \eqref{17}, it follows that $\hat q\, \eta_{x_0}(x)=x_0\eta_{x_0}(x)$. It is instructive to show how this result also follows from a direct computation. Using (\ref{13}) and (\ref{21}) we get
$$\hat q\,\eta_{x_0}(x)=(\1+\alpha  P_{u,v})\hat{q_0}(\1+\beta
P_{u,v})\left(\xi_{x_0}(x)+\alpha\, v(x)\overline{u(x_0)}\right)=$$
$$=(\1+\alpha  P_{u,v})\hat{q_0}(\xi_{x_0}(x)+(\alpha+\beta+\alpha\beta) \overline{u(x_0)}v(x))=$$
$$=
(\1+\alpha  P_{u,v})(x_0\,\xi_{x_0}(x))=\\
x_0\,T\xi_{x_0}(x)=x_0\eta_{x_0}(x).$$
\item $\forall \varphi,\psi \in \SC$
$$
\int_{\Bbb R}dx_0\left<\varphi,\eta^{x_0}\right>\left<\eta_{x_0},\psi\right>=\int_{\Bbb R}dx_0\left<T^\dagger\varphi,\xi_{x_0}\right>\left<\xi_{x_0},T^{-1}\psi\right>=$$
$$=\int_{\Bbb R}dx_0\left(\overline{\varphi}(x_0)+\alpha\overline{\left<v,\varphi\right>}\overline{u}(x_0)\right)
\left(\psi(x_0)+\beta\left<u,\psi\right>v(x_0)\right)=$$
$$=
\left<\varphi,\psi\right>+\left(\alpha+\beta+\alpha\beta\right)\left<u,\psi\right>\left<\varphi,v\right>=\left<\varphi,\psi\right>.
$$
Similarly,
$$
\int_{\Bbb R}dx_0\left<\varphi,\eta_{x_0}\right>\left<\eta^{x_0},\psi\right>=
\left<\varphi,\psi\right>.
$$

\item $\forall x_0,y_0 \in\R$, using \eqref{21}-\eqref{22} and the constraint $\alpha+\beta+\alpha\beta=0$, we have:
$$\left<\eta_{x_0},\eta^{y_0}\right>=\left<\xi_{x_0},\xi_{y_0}\right>+\alpha\overline{u(y_0)}v(x_0)+\beta\overline{u(y_0)}v(x_0)+\alpha\beta\overline{u(y_0)}v(x_0)\left<u,v\right>=\delta(x_0-y_0)$$
We also notice that the condition $\left<\eta_{x_0},\eta^{y_0}\right>=\delta(x_0-y_0)$ is ensured by \eqref{xiproduct}, which can be checked to hold.
In fact,  let $\F=\{e_n(x)\in\Lc^2(\Bbb R)\}$ be an o.n. basis for $\Lc^2(\Bbb R)$, and suppose also that the $e_n(x)$ belongs to $\SC$. Then
we construct the sets $\F_\varphi=\{\varphi_n(x)=T e_n(x)\}$ and $\F_\Psi=\{\Psi_n(x)=(T^{-1})^\dag e_n(x)\}$, where
$$
\varphi_n(x)=e_n(x)+\alpha\,\left<u,e_n\right>v(x),\qquad
\Psi_n(x)=e_n(x)+\overline{\beta}\,\left<v,e_n\right>u(x).
$$
The functions $\varphi_n(x)$ and $\Psi_n(x)$ are all in $\SC$, and it is easy to show that they
form a biorthonormal family, $\left<\varphi_n,\Psi_m\right>=\delta_{n,m}$. Actually, since both $T$ and $T^{-1}$ are bounded, they form two biorthonormal Riesz bases.
Using the expansion  $u(x_0)=\sum_n\left<e_n,u\right>e_n(x_0)$, true in particular for all $u(x)\in\SC,x_0\in \R$, we obtain
$$
\sum_n\left<\xi_{x_0},\Psi_n\right>\left<\varphi_n,\xi_{y_0}\right>=
\sum_n\left<\xi_{x_0},e_n+\overline{\beta}\left<v,e_n\right> u\right>\left<e_n+\alpha\left<u,e_n\right> v,\xi_{y_0}\right>=$$
$$=\sum_n\left[e_n(x_0)+\overline{\beta}\left<v,e_n\right>u(x_0)\right]\left[\overline{e_n(y_0)}+\overline{\alpha\left<u,e_n\right>}v(x_0)\right]=$$
$$=\left(\sum_n e_n(x_0)\overline{e_n(y_0)}\right)+(\overline{\alpha+\beta+\alpha\beta})u(x_0)\overline{v(y_0)}=\left(\sum_n e_n(x_0)\overline{e_n(y_0)}\right)=$$
$$=\delta(x_0-y_0)=\left<\xi_{x_0},\xi_{y_0}\right>.
$$
\end{enumerate}
Conditions (1-3) above ensure that $\F_\eta$ is well behaved. We could further check that, for every $\varphi\in\SC$,  the following relations hold:

$$(\hat q\varphi)(x)=x\varphi(x)+(\alpha \ip{u}{x\varphi}+\beta
\ip{u}{\varphi}x+\alpha\beta\ip{u}{\varphi}\ip{u}{xv})v(x),$$

$$(\hat p\varphi)(x)=-i\frac{d\varphi(x)}{dx}-i\left(\beta \ip{u}{\varphi}\frac{dv(x)}{dx}+\alpha
\left<u,\frac{d\varphi}{dx}\right>v(x)+\alpha\beta\ip{u}{\varphi}\left<u,\frac{dv(x)}{dx}\right>\right)v(x),$$
which give the explicit action of $\hat q$ and $\hat p$ on functions of $\SC$. { In fact, these can be seen as particular cases of the more general situation: let $\Theta_0$ be a self-adjoint operator, mapping $\SC$ into $\SC$ (for instance $\hat q_0$ or $\hat p_0$), and let $T$ as before. Then the operator $\Theta=T\Theta_0T^{-1}$ works on $\SC$ as follows:
$$
(\Theta \varphi)(x)=(\Theta_0\varphi)(x)+(\delta \varphi)(x),
$$
where
$$
(\delta \varphi)(x)=\beta\ip{u}{\varphi}(\Theta_0 v)(x)+\alpha\left[\ip{\Theta_0u}{\varphi}+\beta\ip{\Theta_0u}{v}\ip{u}{\varphi}\right]v(x).
$$
It is interesting to see that, when $\Theta_0$ coincides with $\hat q_0$ or with $\hat p_0$, if $v(x)$ has definite parity, $(\delta \varphi)(x)$ is necessarily not zero. This is easy to see: suppose this is not so, i.e. that $(\delta \varphi)(x)=0$ for all $x\in\Bbb R$. Hence
$$
\beta\ip{u}{\varphi}(\Theta_0 v)(x)=-\alpha\left[\ip{\Theta_0u}{\varphi}+\beta\ip{\Theta_0u}{v}\ip{u}{\varphi}\right]v(x),
$$
which is impossible since the two sides of this equation would have different parities, both if $\Theta_0=\hat q_0$ and if $\Theta_0=\hat p_0$. Hence our map $T$ is non trivial: it really changes the action of $\hat q_0$ and $\hat p_0$, while maintaining the commutation rules between the deformed operators: $[\hat q_0,\hat p_0]=[\hat q,\hat p]=i\1$ (in the sense of unbounded operators).

From a more physical side, we see that $\hat q$ and $\hat p$ differ from their self-adjoint counterparts for an additive term which, in the first case, is a linear combination of $v(x)$ and $xv(x)$, while in the second case, is a linear combination of $v(x)$ and $v'(x)$, with coefficients depending on the function on which these operators are applied. In analogy with the models discussed in recent literature on {\em position-dependent mass}, see for instance \cite{pdm1,pdm2} and references therein, we can call our deformed operators $\hat q$ and $\hat p$ {{\em coordinate-dependent position and momentum operators}}. These operators, when suitably used in the construction of quadratic Hamiltonians of the harmonic oscillator type, give rise to completely solvable models, see for instance \cite{BGV15} for the analysis of this kind of models, even in presence of this explicit dependence on $x$.

}


%
%
%
%
%
%
%
%
%

\subsection{\label{esempio2} Second example}

Let  $T^{-1}$ be the following unbounded operator:
$$T^{-1}:=\1- i(\hat p_0)^2.$$
First of all, it is clear that $T^{-1}$ and its adjoint map $\SC$ into $\SC$. To see that $T$ is  $\SC$-stable, we also have to check that $T$ and $T^\dagger$ do the same. First we need to compute $T$, which can be found by introducing the Green function for $T^{-1}$: $(T^{-1}G)(x)=\delta(x)$. Then, standard computations give
%
%
%
%
%
%
%
%
%
%
%
%
%
%
%

$${G}(x)=\frac{i}{\sqrt{2}(1+i)}e^{-|x|\frac{\sqrt{2}}{2}(1+i)},$$
so that the actions of $T$ and $T^\dag$ on $\varphi(x)\in\SC$   are given by:

\bea
T(\varphi(x))=\frac{i}{\sqrt{2}(1+i)}\int_{\Bbb R}\varphi(x-s)e^{-|s|\frac{\sqrt{2}}{2}(1+i)}ds,\\
T^\dag(\varphi(x))=\frac{-i}{\sqrt{2}(1-i)}\int_{\Bbb R}\varphi(x-s)e^{-|s|\frac{\sqrt{2}}{2}(1-i)}ds.\label{defThelm}\ena
We want to check that $T$ is $\SC$-fully stable. We have already observed that $T^{-1}$ and $(T^{-1})^\dagger$ both map $\SC$ into itself. Less trivial is to check that $T$ also does the same. To see this, we will now check that $x^l \frac{d^k}{dx^k}T(\varphi(x))$ is well defined and goes to zero when $|x|$ diverges, for all $k$ and $l\geq0$.

First we can see that, for all $k\geq0$,
\be
\frac{d^k}{dx^k}T(\varphi(x))=\frac{i}{\sqrt{2}(1+i)}\frac{d^k}{dx^k}\int_{\Bbb R}\varphi(x-s)e^{-|s|\frac{\sqrt{2}}{2}(1+i)}ds =
\frac{i}{\sqrt{2}(1+i)}\int_{\Bbb R}\varphi^{(k)}(x-s)e^{-|s|\frac{\sqrt{2}}{2}(1+i)}ds.
\label{exa1}\en
This can be proved easily since the function $g(s,x):=\varphi(x-s)e^{-|s|\frac{\sqrt{2}}{2}(1+i)}$ satisfies the conditions which ensure the possibility of exchanging integrals and derivatives. In fact,  $\left|\frac{\partial^k g(s,x)}{\partial x^k}\right|\leq M_k e^{-|s|/\sqrt{2}}$, for all $x$, where  $M_k=\sup_{x\in\Bbb R}|\varphi^{(k)}(x)|$, which is finite since  $\varphi(x)\in\SC$. Then, in particular, $T(\varphi(x))$ is a $C^\infty$ function. Of course, from
(\ref{exa1}) we also deduce that
$$
x^l\frac{d^k}{dx^k}T(\varphi(x))=\frac{i}{\sqrt{2}(1+i)}\int_{\Bbb R}x^l \varphi^{(k)}(x-s)e^{-|s|\frac{\sqrt{2}}{2}(1+i)}ds,
$$
for all $k,l\geq0$. Finally, since
\be
\lim_{|x|,\infty} x^l\frac{d^k}{dx^k}T(f(x))=\frac{i}{\sqrt{2}(1+i)}\int_{\Bbb R}\lim_{|x|,\infty} x^l \varphi^{(k)}(x-s)e^{-|s|\frac{\sqrt{2}}{2}(1+i)}ds,
\label{exa2}\en
and since $\lim_{|x|,\infty} x^l \varphi^{(k)}(x-s)=0$ a.e. in $s$, we conclude that $T(\varphi(x))$ belongs to $\SC$. The equality in (\ref{exa2}) follows again from the possibility of exchanging the limit and the integral, which is true because
$$
\left|x^l \varphi^{(k)}(x-s)e^{-|s|\frac{\sqrt{2}}{2}(1+i)}\right|\leq M_{l,k}e^{-|s|/\sqrt{2}},
$$
where $M_{l,k}=\sup_{x,s\in\Bbb R}|x^l\varphi^{(k)}(x-s)|$, which is finite for all $l,k\geq0$, since $\varphi(x)\in\SC$.

Of course, the same holds true for $T^\dagger$, see (\ref{defThelm}). Hence, $T(\varphi(x))\in\SC$.

%
%
%
%
%
%
%
%
%
%
%
%
%
%
%
%

To prove that $T$ is $\SC$-fully stable, it remains to prove that for any sequence  $\varphi_n(x)\in\SC$ $\tau_\Sc$-convergent to $\varphi(x)\in\SC$, then
$(T^{-1}\varphi_n)(x)$ and $(T^\dag\varphi_n)(x)$   are  $\tau_\Sc$-convergent to $(T^{-1}\varphi)(x)$ and to $(T^{\dag}\varphi)(x)$, respectively.
It is clear that this condition is indeed true for $T^{-1}$. Regarding the convergence of $(T^\dag\varphi_n)(x)$, using \eqref{defThelm} we have:
$$\lim_{ n \rightarrow +\infty}x^l \frac{d^k}{dx^k}[(T^\dagger\varphi_n)(x)-(T^\dagger(\varphi)(x)]=$$
$$=\frac{-i}{\sqrt{2}(1-i)}
\lim_{ n \rightarrow +\infty}\int_{\Bbb R}e^{-|s|\frac{\sqrt{2}}{2}(1+i)}x^l\frac{d^k}{dx^k}[\varphi_n(x-s)-\varphi(x-s)]ds=0,$$
due to the Lebesgue dominated convergence theorem.
Hence $T$ is $\SC$-fully stable.

We are now ready to see how the results in Section \ref{sect2} look like in the present case. First of all, by Corollary \ref{cor3},
$\eta_{x_0}\in\SSS, \eta^{x_0}\in\SSS,\quad \forall x_0\in \mathbb{R}$, and that for each $x\in \R$, $\eta_{x_0}(x)\in D(\hat q)$ and $\eta^{x_0}(x)\in D(\hat q^\dag)$. Their explicit expressions is:

\bea
\eta_{x_0}(x)&=&(T\xi_{x_0})(x)=\frac{i}{\sqrt{2}(1+i)}\int_{\Bbb R}e^{-|s|\frac{\sqrt{2}}{2}(1+i)}\delta(x-x_0-s)ds=\\
\nonumber &&=\frac{i}{\sqrt{2}(1+i)}e^{-|x-x_0|\frac{\sqrt{2}}{2}(1+i)},\\
\eta^{x_0}(x)&=&\left((T^{-1})^\dagger\xi_{x_0}\right)(x)=\xi_{x_0}(x)-i\xi_{x_0}''(x)\label{23}.
\ena

We can then prove that $\F_\eta=\{\eta_{x_0}\in\Sc'(\Bbb R), \,x_0\in\Bbb R\}$ is well behaved.

In fact, we first observe that, because of the $\SC-$fully stability condition, Corollary \ref{cor3} implies that $\hat q \eta_{x_0}(x)=x_0\eta_{x_0}(x)$.

\vspace{2mm}

{\bf Remark:--} It is instructive to  verify that $\eta_{x_0}$ is an eigenfunction for $\hat q$ by a direct computation:
$$\hat q\eta_{x_0}(x)=T\hat{q_0}T^{-1}\eta_{x_0}(x)
=T\hat{q_0}T^{-1}\frac{i}{\sqrt{2}(1+i)}e^{-|x-x_0|\frac{\sqrt{2}}{2}(1+i)}=$$
$$
=T\hat{q_0}\left[\frac{ix_0}{\sqrt{2}(1+i)}e^{-|x-x_0|\frac{\sqrt{2}}{2}(1+i)})+
\delta(x-x_0)e^{-|x-x_0|\frac{\sqrt{2}}{2}(1+i)}-\frac{(1+i)x_0}{2\sqrt{2}}e^{-|x-x_0|\frac{\sqrt{2}}{2}(1+i)}\right]=$$
$$
=x_0\frac{i}{\sqrt{2}(1+i)}e^{-|x-x_0|\frac{\sqrt{2}}{2}(1+i)}=x_0\eta_{x_0}(x).
$$

Let us now take $x_0,y_0 \in\R$. Then:

$$\left<\eta_{x_0},\eta^{y_0}\right>=\left<\frac{i}{\sqrt{2}(1+i)}e^{-|x-x_0|\frac{\sqrt{2}}{2}(1+i)},\xi_{y_0}(x)-i{\xi''}_{y_0}(x)\right>=$$
$$\left<\frac{i}{\sqrt{2}(1+i)}e^{-|x-x_0|\frac{\sqrt{2}}{2}(1+i)},\xi_{y_0}(x)\right>+
\left<\frac{i}{\sqrt{2}(1+i)}e^{-|x-x_0|\frac{\sqrt{2}}{2}(1+i)},-i{\xi''}_{y_0}(x)\right>=$$
$$=\left<\frac{i}{\sqrt{2}(1+i)}e^{-|x-x_0|\frac{\sqrt{2}}{2}(1+i)},\xi_{y_0}(x)\right>+$$ $$
\left<\delta(x-x_0)e^{-|x-x_0|\frac{\sqrt{2}}{2}(1+i)}-\frac{(1+i)}{2\sqrt{2}}e^{-|x-x_0|\frac{\sqrt{2}}{2}(1+i)},{\xi}_{y_0}(x)\right>=
\left<\xi_{x_0},\xi_{y_0}\right>=\delta(x_0-y_0),$$
using the distributional derivative. Moreover, $\forall \varphi(x),\psi(x) \in \SC$, we can check that

$$
\int_{\Bbb R}dx_0\left<\varphi,\eta^{x_0}\right>\left<\eta_{x_0},\psi\right>=\int_{\Bbb R}dx_0\left<\varphi,\eta_{x_0}\right>\left<\eta^{x_0},\psi\right>=
\left<\varphi,\psi\right>.
$$

Hence $\F_\eta$ is well behaved, as claimed above.

Moreover it is easy to check that for every $\varphi(x)\in\SC$:

$$\hat q \varphi(x)=T\hat q_0T^{-1}\varphi(x)=\frac{i}{\sqrt{2}(1+i)}\int_{\Bbb R}(x\varphi''(x-s)+ix\varphi(x-s))e^{-|x-s|\frac{\sqrt{2}}{2}(1+i)}ds.$$

$$\hat p \varphi(x)=T\hat p_0T^{-1}\varphi(x)=\frac{i}{\sqrt{2}(1+i)}\int_{\Bbb R}(-i\varphi'''(x-s)+\varphi'(x-s))e^{-|x-s|\frac{\sqrt{2}}{2}(1+i)}ds,$$
which give the explicit expressions of $\hat q$ and $\hat p$ on functions of $\SC$.

{A possible framework where the example of this section could be useful is that of quantum field theory.
In fact, the one-particle Feyman propagator, \cite{mm},
$
D(p_0)=i{\left[{\hat p_0^2}+i(\epsilon+im^2)\1\right]}^{-1}=i{\left[{\hat p_0^2}+iz\1\right]}^{-1},
$
where $\epsilon$ is a constant and $m$ is the particle mass, is equal to operator $T$ introduced in this section when $z=1$.  This suggests the possibility of  studying and analyzing rigorously some mathematical techniques used to circumvent the constraints imposed by the standard formulation of the quantum field theory based on the use of Hermitian operators.

}

\section{Conclusion}\label{sectconcl}

We have seen how two non- self-adjoint {\em position} and {\em momentum} operators, $\hat q$ and $\hat p$, can be analyzed when they are related to the self-adjoint ones, $\hat q_0$ and $\hat p_0$, by some suitable similarity map $T$. In particular, we have shown that biorthogonal eigenvectors can be found for $\hat q$ and $\hat p$, and for $\hat q^\dagger$ and $\hat p^\dagger$ as well, which are distributions in $\SSS$, and which are well-behaved in the sense of \cite{bagriferimento}. We have also discussed in details two examples of somehow different nature, where in particular one can see the explicit form of these eigenvectors. An alternative algebraic settings has also been proposed.

We plan to continue this analysis in a close future, in particular in connection with bi-coherent states, extending what was originally done in \cite{bagriferimento}. We also plan to work more on a physical side of this paper, looking for concrete applications in which the mathematical framework discussed here can be of some utility as, for instance, in the analysis of time-dependent models.

\section*{Acknowledgements}
The authors acknowledge partial support from Palermo University. F.B. and F.G. also acknowledge partial support from G.N.F.M. of the I.N.d.A.M. S. T. acknowledges partial support from G.N.A.M.P.A. of the I.N.d.A.M. F.B. thanks the {\em Distinguished Visitor Program}
 of the Faculty of Science of the University of Cape Town, 2017.

\section*{Computational solution}

This paper does not contain any computational solution.

\section*{Ethics statement}

This work did not involve any active collection of human data.

\section*{Data accessibility statement}

This work does not have any experimental data.

\section*{Competing interests statement}

We have no competing interests.

\section*{Authors' contributions}

FB cured the mathematical part of the paper, with the help of FG, SS and ST. FG, SS and ST cured the concrete applications, with the help of FB. All authors gave final approval for publication.

\section*{Funding}

This work was partly supported by G.N.F.M. and G.N.A.M.P.A.-INdAM and by the University of Palermo.

\end{document}